\documentclass[12pt]{elsarticle}

\usepackage[english]{babel}

\usepackage[letterpaper,top=2cm,bottom=2cm,left=3cm,right=3cm,marginparwidth=1.75cm]{geometry}

\usepackage{amsmath,amsfonts,amssymb,amsthm,euscript,mathrsfs}
\usepackage{graphicx}
\usepackage{algorithm} 
\usepackage{algorithmic}

\newtheorem{theorem}{Theorem}[section]
\newtheorem{lemma}[theorem]{Lemma}
\newtheorem{proposition}[theorem]{Proposition}
\newtheorem{corollary}[theorem]{Corollary}

\newcommand{\calX}{\mathcal{X}}
\newcommand{\calT}{\mathcal{T}}
\begin{document}
%
\begin{frontmatter}

\title{Can Multiple Phylogenetic Trees Be Displayed in a Tree-Child Network Simultaneously?}
\author[label1]{Yufeng Wu}
\address[label1]{Department of Computer Science and Engineering, University of Connecticut, Storrs, CT 06269, USA}
\ead{yufeng.wu@uconn.edu}

\author[label2]{Louxin Zhang}
\address[label2]{Department of Mathematics and Center for Data Science and Machine Learning,\\ National University of Singapore, Singapore 119076\fnref{label4}}
\ead{matzlx@nus.edu.sg}

\begin{abstract}
 A binary phylogenetic network on a taxon set $X$ is a rooted acyclic digraph in which the degree of each nonleaf node is three and  its leaves (i.e.~degree-one nodes) are uniquely labeled with the taxa of $X$.  It is tree-child if each nonleaf node has at least one child of indegree one. A set of binary phylogenetic trees may or may not be simultaneously displayed in a binary tree-child network. Necessary conditions for multiple phylogenetic trees being simultaneously displayed in a tree-child network are given here. In particular, it is proved that any two phylogenetic trees can always simultaneously be displayed in some tree-child network on the same taxa set.  It is also proved that any set of multiple binary phylogenetic trees can always simultaneously be displayed in some non-binary tree-child network  on the same taxa set, where each nonleaf node is of either indegree one and outdegree two or indegree at least two and outdegree out.
\end{abstract}
\end{frontmatter}

\section{Introduction}
Phylogenetic trees have been used to model the evolutionary history of species for several hundred years \cite{Fel04Book2}. Recent genetic and genomic studies suggest that 
horizontal evolutionary events have played more important roles  than we expected in genome evolution \cite{Fontaine_15,Marcussen_14,Moran_10}. As such, phylogenetic networks, which are rooted acyclic digraphs, have been used more and more to model the evolutionary history of genomes with the presence of recombination, hybridization and other reticulate events \cite{Gusfield_book,Huson_book}.  

Given that the space of phylogenetic networks is huge, it is extremely hard to design fast exact algorithms to reconstruct phylogenetic networks from bio-molecular sequence data \cite{elworth2019}. Since each phylogenetic network can be considered as the consensus of multiple phylogenetic trees, one approach is to infer a phylogenetic network from phylogenetic trees. That is, given a set of multiple phylogenetic trees, to infer a phylogenetic network that displays the given trees simultaneously. 
 Mathematical relation between tree and networks have extensively been studied in recent years (\cite{Steel_book,Zhang_18} for a survey).  



The network inference problem has been extensively studied by researcher in the bioinformatics community. There are two types of approaches for this problem: unconstrained network construction and constrained network construction. Unconstrained network construction \cite{WURN10,ZW2012,WURNJCB13,WUMIRRN16} attempts to construct a network without additional topological constraints. The downsides of these approaches include that the constructed network contain a unreasonable number of reticulate events and that the computer programs are slow even for trees on a few dozen taxa. Constrained network reconstruction imposes some kind of topological constraints, which can simplify the network structure and often lead to more efficient algorithms. There are various kinds of constraints studied in the literature. One constraint is requiring simplified cycle structure in networks \cite{Gusfield_04,Wang_01}.
Another is the so-called tree-child property \cite{Cardona_09b}.

A binary phylogenetic network is tree-child if every nonleaf node has at least one child that is of indegree one, which implies that every nonleaf node is connected to some leaf through a series of tree edges. 
Although tree-child networks can have complex structure, they can efficiently be enumerated and counted  by a simple recurrence formula \cite{pons2021_SR,Zhang_19}. In addition, the problem of determining whether or not a tree is displayed in a phylogenetic network is NP-complete \cite{Kanji_08} in general but becomes linear time solvable for tree-child networks \cite{Gunawan_16_IC,van2010_IPL} and so is whether or not a phylogenetic network is displayed in another \cite{janssen}.  A parametric algorithm has also been developed for determining whether a set of multiple trees can be display in a tree-child network simultaneously  \cite{van2022practical}.  As an exponential-time algorithm,  it may become slow for large data.

In this paper, we focus on characterizing  multiple phylogenetic trees that can be simultaneous displayed in binary tree-child networks. The computational complexity of determining the simultaneous display of multiple trees for tree-child networks is an open question \cite{report_2020}. Here, we present necessary conditions for phylogenetic trees to be simultaneously displayed in a tree-child network in terms of cherry count (Section~\ref{sec3}) and then prove that every pair of trees can always be  displayed in a tree-child network simultaneously (Section~\ref{sec4}). To the best of our knowledge, such mathematical quantification of simultaneous display of multiple trees in a tree-child network has yet to be rigorously investigated in the literature. Our results reveal new properties of tree-child network that may be useful for developing efficient algorithms for determine whether or not a set of trees can be simultaneously displayed in a tree-child network.  Although a set of phylogenetic trees may not be simultaneously displayed in any binary tree-child network, we also prove that any set of phylogenetic trees can be simultaneously displayed in some non-binary phylogenetic network in which a reticulate node may have indegree greater than two in Section~\ref{sec5}.

\begin{figure}[t]
\label{Fig1_examples}
\centering
\includegraphics[width=0.7\textwidth]{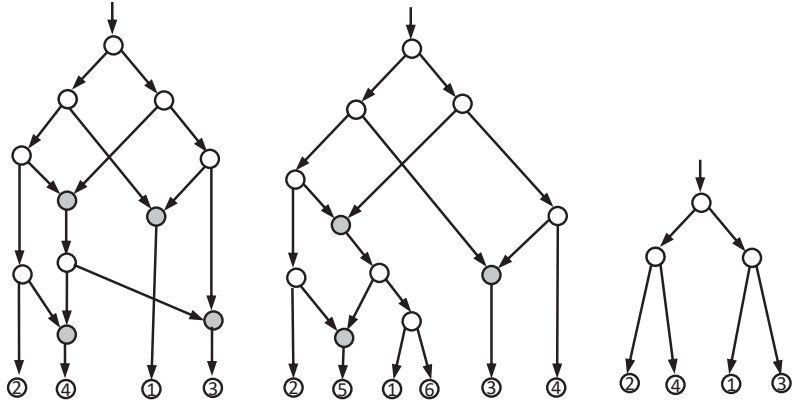}
\caption{An arbitrary phylogenetic network with four reticulate nodes (filled circles) on taxa 1, 2, 3, 4 (left), a tree-child network (middle) and a phylogenetic tree (right). }
\end{figure}

\section{Preliminaries}

\subsection{Phylogenetic networks}
A {\it phylogenetic network} on a set of taxa is a rooted acyclic digraph in which
there is a directed path from the root to every other node, the root is of outdegree two, each non-root node is of degree three or degree one, and degree-one nodes (i.e. leaves) represent  each a unique taxon.
In this paper, for convenience of study,  a branch is added right above its root, called the {\it root edge}.  A phylogenetic network on 
$\{1, 2, 3, 4\}$ is depicted in 
Figure~\ref{Fig1_examples} (left).

Let $\mathcal{X}$ be a set of taxa and $N$ be a phylogenetic network on  $\mathcal{X}$. A node  of $N$ is called a {\it reticulate node} if it is of indegree two and outdegree one.  A node is called a {\it tree node} if it is not a reticulate node.  Note that a tree node is the root, a leaf or a node of indegree one and outdegree two.

An edge of $N$ is called a {\it tree edge} if it enters a tree node.  It is called a {\it reticulate edge} if it enters a reticulate node.

A phylogenetic network is {\it tree-child} if every nonleaf node has a child that is a tree node.  The middle phylogenetic network in Figure~\ref{Fig1_examples} is tree-child. In contrast, the left network in Figure~\ref{Fig1_examples} is not tree-child, where both the parent $u$ of Leaf 4 and the parent $v$ of Leaf 3 are reticulate and the node right above $u$ has $u$ and $v$ as its child.     Clearly, tree-child networks have the following property. 

\begin{proposition}
\label{Prop1_TCproperty1}
 Let $N$ be a phylogenetic network. If $N$ is a tree-child, then for every non-leaf node $v$ of $N$, there is a directed path $P$ from $v$ to some leaf $\ell$ that consists of only tree edges.
\end{proposition}

Let $u$ and $v$ be two nodes of $N$. We say that $u$ is a {\it parent} of $v$ or $v$ is a {\it child}  of $u$ if $(u, v)$ is an edge of $N$. In this paper, we use $c(u)$ to denote the unique child of a reticulate node and $p(v)$ denote the unique parent of a tree node that is not a root.  More generally, we say that $u$ is an {\it ancestor} of $v$ or $v$ is a {\it descendant} of $u$ if there is a direct path from $u$ to $v$. 

A node is said to be a {\it common ancestor} of $u$ and $v$ if it is an ancestor of both $u$ and $v$. The {\it least common ancestor} of $u$ and $v$ is defined to be their common ancestor that is a descendant of any other common ancestors of $u$ and $v$, denoted by $\mbox{\rm lca}_N(u, v)$ if it exists. Note that the least common ancestor always exists for any pair of nodes in a phylogenetic tree. 

\subsection{Network decomposition}

Let $N$ be a phylogenetic network with $k$ reticulate nodes. 
Let the root of $N$ be $\rho$ and the $k$ reticulate nodes be 
$r_1, r_2, \cdots, r_k$. 
For each $x \in \{\rho, r_1, r_2, \cdots, r_k\}$, $x$ and its descendants that are connected to $x$ by a path consisting of only tree edges induces a subtree of $N$. Such $k+1$ subtrees are called the {\it tree components} of $N$
in \cite{Gunawan_16_IC}. It is easy to see that the tree components are disjoint and the node set of $N$ is the union of these tree components, as shown in Figure~\ref{Fig2_decomposition}.  This network decomposition is a powerful technique for studying the combinatorial properties and algorithms for tree-child networks \cite{Cardona_20_JCSS,fuchs_21_JEC} and other network classes \cite{Gambette_15} (see \cite{Zhang_18} for a survey).

We deduce the following known fact about tree-child networks from Proposition~\ref{Prop1_TCproperty1}. 

\begin{proposition}
\label{Prop2_retNum}
(\cite{Cardona_09b}) Let $N$ be a tree-child network on $n$ taxa. Each tree component of $N$  contains at least one leaf.  Therefore, $N$ contains $n-1$ reticulate nodes at most.
\end{proposition}

A tree component is said to be {\it complex} if it contains more than one leaf. 

\begin{proposition}
\label{Prop3_componentStruct}
Let  $N$ be a tree-child network on $n$ taxa and $C$ be a tree component of $N$. If $C$  contains $k$ leaves (of $N$), then $C$ contains $k-1$ nodes whose children are both in $C$. In particular, a non-complex tree component is a top-down path to a leaf from either the network root or a reticulate node. 
\end{proposition}

The network depicted in  Figure~\ref{Fig2_decomposition}
have two complex and two non-complex tree components.

\begin{figure}[t]
\centering
\includegraphics[width=0.3\textwidth]{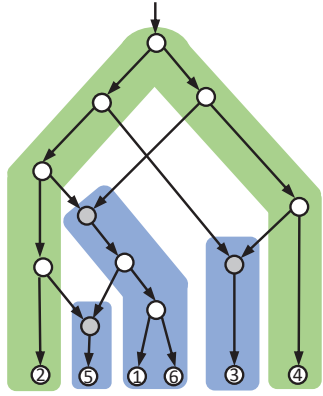}
\caption{Illustration of the decomposition of a phylogenetic network with $k$ reticulate nodes into $k+1$ disjoint tree components, where $k=3$.  The  tree component rooted at the network root is highlighted in green; other tree components each rooted at a recirculate node are in blue.
\label{Fig2_decomposition}}
\end{figure}

\subsection{Phylogenetic trees}

A phylogenetic tree is a phylogenetic network with no reticulate nodes,
as shown in Figure~\ref{Fig1_examples} (right). 
Let $T$ be a phylogenetic tree on ${\cal X}$ and $n=\vert {\cal X}\vert$. Then, $T$ has $n-1$ nonleaf nodes and $2n-1$ edges including the open root edge. 

Let $u$ be a node of $T$. We define the {subtree rooted at} $u$, denoted by $T_u$,  to be the tree with 
the node set $D(u)$ consisting of $u$ and all its descendants and the edge set consisting of the edges between the nodes of $D(u)$. 

A {\it cherry} of $T$ is a subtree rooted at a node whose children are both leaves. For example, the tree given in Figure~\ref{Fig1_examples} contains two cherries. One consists of 1 and 3, while another consists of 2 and 4.  It is not hard to see that every phylogenetic tree has at least one cherry and at most $n/2$ cherries. For any pair of taxa $a$ and $b$, we use $(a, b)$ to denote the cherry with leaves labeled with $a$ and $b$. Because there is no edge between any pair of two leaves, no confusion will arise with this cherry notation.

\subsection{Trees displayed in tree-child networks}

Consider a tree-child network $N$ on ${\cal X}$. There are two edges entering each reticulate nodes. If we remove only one entering edge for every reticulate node, the resulting subgraph $N'$ satisfies that  
each non-leaf node is of indegree one.
Note that $N'$ may
contain unlabeled leaves (i.e. nonleaf nodes of $N$ whose outgoing edges were all removed) if $N$ is not tree-child. Therefore,
we obtain a phylogenetic tree $T$ on $\calX$ from $N'$ if we further remove all the
nodes that are not the ancestors of any labeled leaves and then contract all the degree two
nodes.
The phylogenetic trees that can be obtained in this way are said to be a {\it displayed} tree of $N$. For instance, the tree in Figure~\ref{Fig1_examples} (right) is a tree displayed in the left network in the figure, as demonstrated in Figure~\ref{Fig3_treeDisplay}.

A set of multiple phylogenetic trees are said to be {\it co-displayed} in  a tree-child network if the network displays all the given trees simultaneously.

If $N$ contains $k$ reticulate nodes, there are $2^k$ possible ways for deleting only one entering edge for each reticulate node. As such, $N$ can display at most $2^k$ distinct phylogenetic trees. In particular, Proposition 
~\ref{Prop2_retNum} implies the following result.

\begin{proposition}
\label{propCount}
Let $N$ be a phylogenetic network on ${\cal X}$.  If $N$ is tree-child, it can display at most $2^{n-1}$ distinct phylogenetic trees, where $n=\vert{\cal X}\vert$.
\end{proposition}

Throughout the paper, when we refer to trees, we mean phylogenetic trees. We also simply call phylogenetic tree-child networks tree-child networks. 
In the rest of the paper, we focus on the problem of determining whether a set of trees on ${\cal X}$ can be co-displayed in a tree-child network on ${\cal X}$ or not.

\begin{figure}[t]
\centering
\includegraphics[width=0.8\textwidth]{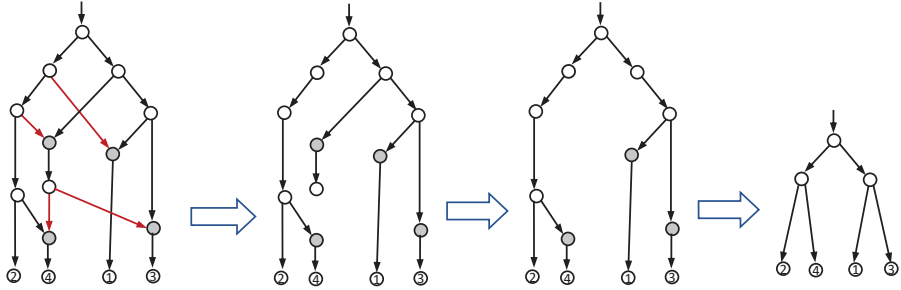}
\caption{Illustration of removing (red) entering reticulate edges (Step 1), deleting nodes such that there is no labeled leaf among its descendants (Stpe 2) and contracting degree-two nodes (Step 3).    
\label{Fig3_treeDisplay}}
\end{figure}

\section{Necessary conditions for co-display of multiple trees}
\label{sec3}

Let ${\cal X}$ be a set of $n$ taxa.
There are  
$a_n=1\times 3\times \cdots (2n-3)$ possible trees on ${\cal X}$. 
Note that $a_n> 2^{n-1}$ for $n\geq 4$. 
By Proposition \ref{propCount}, any more than $2^{n-1}$
distinct trees on $\calX$ cannot be co-displayed in any tree-child network on $\calX$ for $n\geq 4$. 
In practice, however, we  often have a small number of trees and would like to know whether  they can be co-displayed in a tree-child network or not.
In this section, we present necessary conditions for the co-display of multiple trees in tree-child networks in terms of the structural properties of trees.

\begin{figure}[bh!]
\centering
\includegraphics[scale=1.0]{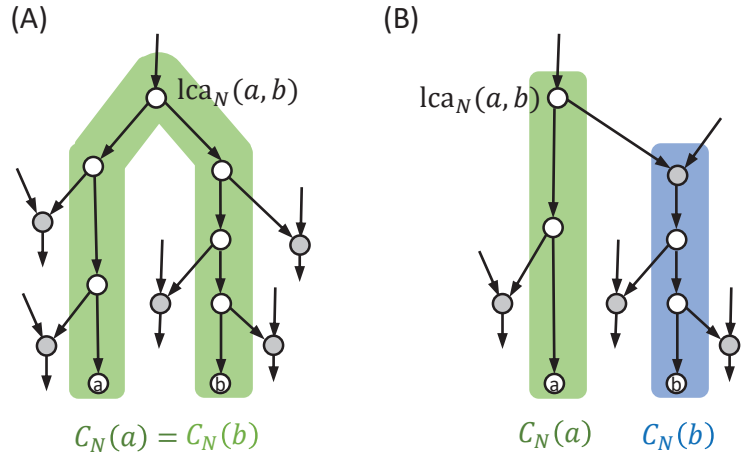}
\caption{Illustration of the display of a cherry in a tree-child network. (A) The two leaves of a displayed cherry are in the same component.  (B) The two leaves of a displayed cherry are in different tree components.  
\label{Fig4_cherry_display}}
\end{figure}

\subsection{Cherry-based necessary conditions for co-display of multiple trees}

Consider a set $\calT$ of $m$ trees on $\calX$. 
The number $c(\cal T)$ of distinct cherries that appear in at least one given tree is called the {\it cherry number} of $\cal T$.  For example, for $$\mathcal{T} = \{ (((1,2),(3,4)),5),  (((1,2),(3,5)),4) \},$$
$c(\calT)=3$ as only cherries $(1,2), (3,4), (3,5)$ appears in the trees of $\mathcal{T}$,
where the two trees are given in the Newick format, in which each pair of parentheses surround a node cluster of the represented tree.

\begin{theorem}
\label{theoremCherry}
Let $\calT$ be a set of trees on $n$ taxa and $c(T)$ be its cherry number.  If  $\mathcal{T}$ are co-displayed in a tree-child network on the same taxa, then 
$c(\calT)\le 2n-3$.
\end{theorem}


Let $N$ be a tree-child network on a set $\calX$ of $n$ taxa. 
A cherry $(a, b)$ is said to be displayed in $N$ if it is a subtree of the network $N'$ obtained from $N$ through the removal of a set of reticulate edges that enter distinct reticulate nodes and contraction of degree-2 nodes.  Clearly, $c({\cal T})$ is less than or equal to the number of cherries that are displayed in $N$ if $N$ co-displays ${\cal T}$. We will prove the above theorem by proving that the number of cherries displayed in a tree-child network` on $\calX$  is at most $2n-3$. To this end, we first establish the following fact.

\begin{lemma}
\label{lemma1}
Let $N$ be a tree-child network on $\calX$ and $x, y\in \calX$ and let $C_N(z)$ denote the tree component that contains $z$ for any tree node $z$ of $N$. If the cherry $(x, y)$ is displayed in $N$, then one of the following statements holds.
\begin{itemize}
    \item  $C_N(x)=C_N(y)=C_N(\mbox{\rm lca}_N(x, y))$. In addition,  every middle node in the path from $\mbox{\rm lca}_N(x, y)$ to $x$ (resp. $y$) is of degree two in $C_N(x)$, as shown in Figure~\ref{Fig4_cherry_display}A.
    \item   $C_N(\mbox{\rm lca}_N(x, y))= C_N(x)\neq C_N(y)$.
    In addition, $C_N(y)$ is a path from its root $r_y$ to the leaf $y$ and $(\mbox{\rm lca}_N(x, y), r_y)$ is an entering edge of $r_y$,  as shown in Figure~\ref{Fig4_cherry_display}B. 
    \item   $C_N(\mbox{\rm lca}_N(x, y))= C_N(y)\neq C_N(x)$.
    In addition, $C_N(x)$ is a path from its root $r_x$ to the leaf $x$ and $(\mbox{\rm lca}_N(x, y), r_x)$ is an entering edge of $r_x$. 
\end{itemize}
\end{lemma}
\begin{proof} Let $(x, y)$ be a displayed cherry of $N$.
 We consider two possible cases.
 
 {\bf Case 1}. The leaves $x$ and $y$ are the two leaves in a complex tree component $C$. 
 
 Let the root of $C$ is denoted by $\rho_C$.  Since $C$ is a rooted tree, $\mbox{lca}_N(x, y)=\mbox{lca}_C(x, y)$ exists in $C$, denoted by $\mbox{\rm lca}_{xy}$.  Let $P$ be the path from $\mbox{\rm lca}_{xy}$ to $x$.
 If $P$ contains a node $a$ that has two children $b$ and $c$ in $C$ such that $b$ is not in $P$, by Proposition~\ref{Prop1_TCproperty1}, $b$ is connected to a leaf $z\neq x$ via a directed path in $C$. The removal of any set of reticulate edges will never remove the nodes in $C$ and thus it is impossible to obtain a cherry $(x, y)$ in the resulting network.
 Therefore, each middle node of the path $P$ is of degree-two in $C$. By the same argument, each middle node of the path from  $\mbox{\rm lca}_{xy}$ to $y$ is of degree-two in $C$. This case is shown in Figure~\ref{Fig4_cherry_display}A.
 
  {\bf Case 2}. The leaves $x$ and $y$ are in different complex tree components,
  i.e. $C_N(x)\neq C_N(y)$. 
  
  We use $\rho_a$ and $\rho_b$ to denote the root of $C_N(a)$ and $C_N(b)$ respectively. Since $N$ is acyclic, there are three possibilities: 
  (i) only $\rho_y$ is a descendent of $\rho_x$,  (ii) only $\rho_x$ is a descendent of $\rho_y$ and (iii) 
  $\rho_x$ and $\rho_y$ has no  ancestor-descendant relation.
  
  If Case (i) holds, then there is a path $P$ from $r_x$ to $r_y$. 
  If $P$ passes through other tree-components, it is impossible for $x$ and $y$ to form a cherry that is displayed in $N$, as every tree component contains at least one leaf.  The path $P$ leads to $r_y$ through an entering edge $e$ of $r_y$, as the parent of every node other than $r_y$ of $C_N(y)$ is in $C_N(y)$ (Figure~\ref{Fig4_cherry_display}B). This also implies that $y$ must be the unique leaf of $C_N(y)$ and that the start node of $e$ is $\mbox{\rm lca}_{xy}$. We have proved that the point two in the lemma holds.
  
  Case (ii) is symmetric to Case (i). By the same reasoning, the point three in the lemma holds.
  
  Case (iii) is impossible. The reason is that under such a case, any undirected path between $x$ and $y$ must go through other tree components and thus $x$ and $y$ cannot form a cherry that is displayed in $N$.  This concludes the proof.
\end{proof}


Consider a displayed cherry $(a, b)$ of $N$.  It is said to be {\it type-1} if both $a$ and $b$ belong to the same complex tree component. It is {\it type-2} if $a$ and $b$ belong to distinct neighboring tree components, respectively.
\\

{\it Proof of Theorem~\ref{theoremCherry}.}
Let $t$ be the number of the leaves of $N$ that appear in all complex tree components. Then, there are $n-t$ tree components that each  contains exactly one leaf, which are simply a path from their root to their leaf. 

For a complex tree component $C$ with $k$ leaves. If $a$ and $b$ are two leaves of $C$ that form a type-1 displayed cherry, by Lemma~\ref{lemma1}, there is  no degree-three nodes in the paths from $\mbox{lca}_C(a,b)$ to $a$ and $b$. This implies that each leaf of  $C$ can appear in at most one type-1 displayed cherry and thus all the $k$ leaves of $C$ can give  at most $k/2$ type-1 displayed cherries. Overall, $N$ can display $t/2$ type-1 cherries at most.

Each type-2 displayed cherry $(a, b)$ contains  at least one leaf that appears in a non-complex tree component. By Lemma~\ref{lemma1}, if $\mbox{lca}_N(a, b)$ appears in the same tree component as $a$, which is denoted by $C_N(a)$, then $b$ must be the unique leaf of the non-complex tree component $C_N(b)$ and the reticulate edge  connects $\mbox{lca}_N(a, b)$  (in $C_N(a)$) to the root of $C_N(b)$. This implies that for each non-complex tree component $S$, its unique leaf $\ell_S$ can appear in at most two type-2 displayed cherries $(x, \ell_S)$ such that $\mbox{lca}_{N}(x, \ell_S)$ and $x$ are both in $C_N(x)$. Therefore, there are at most 
$2(n-t)$ displayed cherries that contains the unique leaf of some non-complex tree component. 

If $t\neq 0$, then
$t\geq 2$ and thus
the number of cherries that are displayed in $N$ is at most
$t/2+2(n-t)=2n-3t/2 \leq 2n-3$.

If $t=0$, there are $n$ tree components each containing a unique leaf and all displayed cherries are of type-2.  Let the $n$ tree components of $N$ be 
$C_0, C_1, \cdots, C_{n-1}$, where $C_0$ is the one that contains the network root and the rest are rooted at a reticulate node and contain a unique leaf (Figure~\ref{fig:ExNet1}).
We further use  $\ell_i$ and $\rho_i$ to denote the unique leaf and the root of $C_i$, respectively, for each $i$.

For an $i>1$,  $\rho_i$ is a reticulate node and has two parents, say $p'_i$ and $p''_i$. If the two parents are in distinct tree-components $C_j$ and $C_k$, $j\neq k$, then
$(\ell_j, \ell_i)$ and
$(\ell_k, \ell_i)$ are the only two possible cherries which can be displayed and whose display involve a reticulate edge entering $\rho_i$. If $p'_i$ and $p''_i$ are found in a common tree component $C_j$, then,
only $(\ell_j, \ell_i)$ can be a displayed cherry whose display involves a reticulate edge entering $\rho_i$.
Since $N$ is acyclic, there are at least one reticulate node whose parents are both in $C_0$, as shown in Figure~\ref{Fig2_decomposition}.  Counting the displayed type-2 cherries in terms of the used reticulate edges, we have that the number of the cherries displayed in $N$ is at most
$2(n-1)-1=2n-3$.\\
\hspace*{\fill}$\Box$

\begin{figure}[t]
\centering
\includegraphics[scale=0.8]{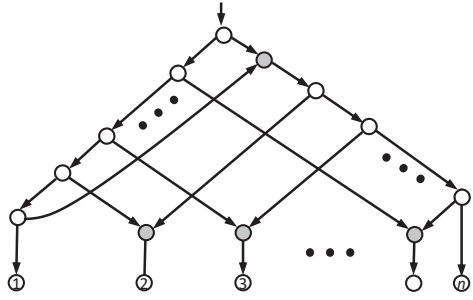}
\caption{A tree-child network with $n$ taxa that displays $2n-3$ cherries $(1, i), (1, n), (i,n)$,  $i=2, ..., n-1$. \label{fig:ExNet1}}
\end{figure}

\noindent {\bf Remark.} 
The bound of $2n-3$ is tight, as  the tree-child network on $n$ shown in Figure~\ref{fig:ExNet1} displays exactly $2n-3$ cherries: $(1,i),  (1,n), (i,n)$ for each $i=2, \cdots, n-1$.

\begin{corollary}
\label{corollary34}
For odd $n\geq 5$, there exist four trees on $n$ taxa that cannot be co-displayed in a tree-child network.
\end{corollary}

\begin{proof}
%
For odd $n$, we arrange 
the taxa 1 to $n$ on the vertices of the regular $n$-sided polygon $G$ clockwise (left panel, Figure~\ref{Fig6_4trs}). For each side $(i, i+1)$ or $(n, 1)$ of $G$, there are $\frac{n-1}{2}-1$ diagonals parallel to it. 
We can use such  $\frac{n-1}{2}$ parallel side/diagonals to define a tree with $\frac{n-1}{2}$ cherries. More precisely, we can construct  four trees $T_1$-$T_4$ such that:
\begin{itemize}
    \item $T_1$ contains the $(n-1)/2$ cherries 
    $(1, 2), (n, 3), (n-1, 4), ..., \left(\frac{n+5}{2}, \frac{n+1}{2}\right)$.
    \item $T_2$ contains the $(n-1)/2$ cherries 
    $\left(\frac{n-1}{2}-j, \frac{n+1}{2}+j\right)$, $j=0, 1, \cdots, \frac{n-3}{2}$.
    \item $T_3$ contains the $(n-1)/2$ cherries 
    $\left(\frac{n+1}{2}-j, \frac{n+3}{2}+j\right)$, $j=0, 1, \cdots, \frac{n-3}{2}$.
    \item $T_4$ contains the $(n-1)/2$ cherries 
    $(n, 1), (n-1, 2), ..., \left(\frac{n+3}{2}, \frac{n-1}{2}\right)$.
\end{itemize}
In total, the four trees $T_1$-$T_4$ contain  $2n-2$ distinct cherries, which is larger than $2n-3$. By Theorem~\ref{theoremCherry}, 
the trees cannot be co-displayed in any tree-child network.
\end{proof}

\begin{theorem}
\label{theorem2Cherry_2}
Suppose a set of trees $\mathcal{T}$ on an $n$-taxa set $\mathcal{X}$ are co-displayed in a tree-child network.  
Then,  there exists a taxon $x$ such that at most two distinct cherries consisting of $x$ and another taxon appear in the given trees.
More generally, for any integer $k$ $(1 \le k \le n-1)$, the number of taxa $x$ that appears in $k$ or less distinct cherries in the trees is at least $k-1$.
\end{theorem}
\begin{proof} 
Since $N$ is acyclic, there must be a tree component $C$ that does not contain the parents of any reticulate node. If $C$ contains more than two leaves, there is at least one cherry.
Let $(x, y)$ be a cherry appearing in 
$C$. The cherry $(x, y)$ is then the unique cherry containing $x$ in $N$. 

If $C$ contains exactly two nodes: its root $r_C$ and its unique leaf $\ell_C$. For example, the network in Figure~\ref{fig:ExNet1} contains $n-2$ such two-node tree components in the middle of the bottom. It is easy to see that $\ell_C$ can only form a displayed cherry with the leaf of
a tree component that contains one of two parents of $r_C$. Since there are at most two
such tree components, the statement is tree for $\ell_C$. 

More generally,  we can arrange  the  components of $N$ in the topological order (such that each reticulate edge starts from a component and enters a component on the right) and number them from $1$ to $c$ ($c\leq n-1$)  starting from the left.
We consider the last $k-1$ tree components. For a leaf $\ell$ in a non-complex component, it can form a cherry with at most two leaves in the tree components on the left and at most $k-2$ leaves on the right. 
For a leaf $\ell$ in a complex component, it can form at most $k-2$  cherries with 
a leaf in the component on the right and at most one cherry  with a leaf within the component. 
\end{proof}

\begin{corollary}
\label{corollary3Trees}
For even $n\geq 4$, there exist three trees on $n$ taxa that cannot be co-displayed in a tree-child network.
\end{corollary}

\begin{proof}
For $n=4$, we consider the following three trees: $$T_1=((1,2),(3,4)),  T_2=((1,3),(2,4)), T_3=((1,4),(2,3))$$ on taxa $\{1, 2, 3, 4\}$. 
Since every taxa appears in three distinct cherries in the given trees, by Theorem~\ref{theorem2Cherry_2}, we conclude that there is no tree-child network on the taxa that can display the three trees simultaneously.

For even $n\geq 6$, we arrange $1$ to $n-1$ to the vertices of 
the regular $(n-1)$-sided polygon as described above and
place the vertex $n$ in the center of the polygon (right panel, Figure~\ref{Fig6_4trs}).
For each side $(i, i+1)$ ($1\leq i\leq n-2$) or $(n-1, 1)$, there are $(n/2-2)$ diagonals parallel  and one diagonal (incident to the vertex $n$) perpendicular to it.  We can use this fact to define
three trees each containing a disjoint set of $n/2$ cherries. Since each taxa appears in three distinct cherries in the constructed trees,  by Theorem~\ref{theorem2Cherry_2}, the three trees cannot be co-displayed in any tree-child network. 
\end{proof}

\begin{figure}[b!]
\centering
\includegraphics[scale=0.6]{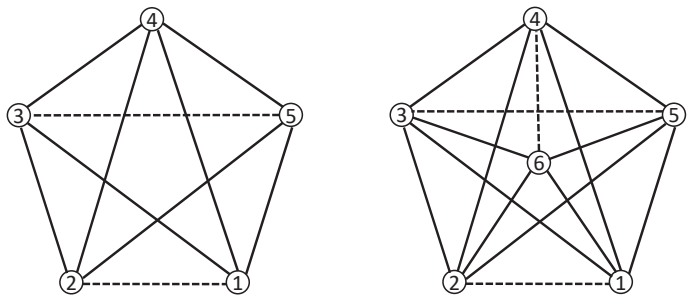}
\caption{For odd $n$,  in 
the regular $n$-sided polygon, there are  $(n-3)/2$ diagonals parallel to every side.  The sets of $(n-1)/2$ parallel sides/diagonals are used to define trees with $(n-1)/2$ cherries in the proof of Corollary~\ref{corollary34}. For even $n$, the sets of $(n-2)/2$ parallel side/diagonals in the the regular $(n-1)$-sided polygon plus a diagonal incident to $n$ is used to define trees with $n/2$ cherries in the proof of Corollary~\ref{corollary34}. Here $n=5$ and the dashed sides/diagonals are used to define a tree. \label{Fig6_4trs}}
\end{figure}

\subsection{Generalizations}

\subsubsection{Subset of taxa}

We now present a stronger  necessary condition for the co-display of multiple trees on  $\calX$ by considering the cherries consisting of taxa drawn from the \emph{subsets} of  $\mathcal{X}$. 
We need the following lemma for establishing a stronger necessary condition.

\begin{lemma}
\label{lemmaSingleton}
For any tree-child network $N$ on an $n$-taxa set $\calX$, we can \emph{add} new reticulate edges to $N$ 
such that $N$ remains a tree-child network after insertion and decomposition of this new network leads to all singleton components. That is, each tree-child network is displayed in at least one tree-child network with $n-1$ reticulations on $X$.
\end{lemma}

\begin{proof}
Suppose $N$ has one or more complex tree components. We add reticulate edges to make each of such components to become a union of multiple non-complex tree components using the following process.
 Each complex component $C$ has a root node $r_C$ and has internal nodes $v$ that have two children $v_1$ and $v_2$, where $v$ is selected arbitrarily. We create a reticulate edge  $(r_C, c(r_C))$ into an outgoing edge of $v$, say $(v, v_1)$. Note that this would make $v_1$ into the child of the root of a new component, which consists of the parent of $v_1$, $v_1$ and the descendants of $v$ in $C$. 
 Note that $v_1$ remains as  a tree node. Also, adding the reticulate edge don't create any cycle since  $r_c$ is an ancestor of $v$ in $N$.  Since $v_2$ remains as a tree-node and as a child of $v$, the resulting network is tree-child but has one more tree component. 
We repeat this procedure until there is no complex tree component. Since the tree-child property is kept, the procedure produces a tree-child network with no complex tree components.
\end{proof}

Using the above lemma, we are able to prove the following necessary condition that is stronger than one in Theorem \ref{theoremCherry}.

\begin{theorem}
\label{theoremGenCherry}
Suppose a set of trees $\mathcal{T}$ on an $n$-taxa set $\mathcal{X}$ are co-displayed in a tree-child network.  Let $c_{\calX'}(\calT)$ denote the number of cherries that consist of taxa in $\calX'$ for a subset $\calX '\subseteq \calX$.
Then  for every subset $\mathcal{X}' \subseteq \mathcal{X}$ such that $\vert \calX'\vert \geq 2$,
    $c_{\calX'}(\calT)\leq 2\vert \calX'\vert -3$.
\end{theorem}


\begin{proof} 
Let the trees of $\cal T$ be co-displayed in a tree-child network $N$. Without loss of generality, by Lemma~\ref{lemmaSingleton}, we may assume that $N$ contains $n-1$ reticulate nodes and all its tree components contain only one taxa.

Recall that we assume that $N$ has $n$ singleton tree components and has $n-1$ reticulate nodes. 
We now count the distinct displayed cherries consisting only of the taxa of $\mathcal{X}'$. 
Note that cherries can only be formed by connecting two of $n$ singletons.
To have as many cherries on $\calX'$,  we should connect singletons
with taxa from $\mathcal{X}'$ as much as possible.
More specifically, if a reticulate edge of a singleton with leaf $x \in \mathcal{X}'$
is not originated from another singleton with leaf $x' \in \mathcal{X}'$, then
we simply change the source of this reticulate edge to be from some taxon in $\mathcal{X}'$. This won't reduce the number of cherries on $\calX'$.
We follow the exactly the same reasoning as the proof in Theorem \ref{theoremCherry}:
the singletons from $\mathcal{X}'$ cannot form cycles and so have a topological order; each singleton can create at most two distinct cherries in $\mathcal{X}'$;
 the first singleton in the topological order cannot form cherries within $\mathcal{X}'$ (note that it is possible this first singleton can form cherries with taxa outside $\mathcal{X}'$ but these don't affect the number. 
and the second singleton in this order can only form one cherry within $\mathcal{X}'$.
Therefore,  for any $\mathcal{X}' \subseteq \mathcal{X}$,
 $  c_{\calX'}(\calT) \le    2\vert {\cal T'}\vert  -3    $
\end{proof}


Theorem \ref{theoremGenCherry} can be used to no co-display of multiple trees  when Theorem \ref{theoremCherry} fails.  For example, we consider the tree set $\calT$ consisting of the following six trees on five taxa:
\begin{quote}
 $((((1,2),3),4),5), \;\; ((((1,3),2),4),5), \;\;
 ((((1,4),5),2),3),$  \\
 $((((2,3),1),4),5), \;\;
  ((((2,4),5),1),2), \;\;
((((3,4),1),2),5).$
\end{quote}
Since $c_{\calX}(\calT)=6 < 2\times 5-3$, $\calT$ satisfies the necessary condition given in Theorem \ref{theoremCherry}. However, 
for the subset consisting of the taxa $1,2,3$ and $4$, we have
$c_{\{1, 2, 3, 4\}}(\calT)=6> 5=2\times 4-3$, implying that the six trees cannot be co-displayed in a tree-child network according to Theorem \ref{theoremGenCherry}. 
\\

\noindent {\bf Remarks}
1.) We emphasize that the necessary condition in Theorem \ref{theoremGenCherry} is for co-display of multiple trees in a tree-child network. It is not for co-displaying the trees as the binary spanning trees on a subset of taxa.
In fact, 
it is known that any set of trees on $\calX$ can be displayed in a phylogenetic network on $\calX$ (see \cite{Francis_15,Zhang_16} for example). 
It is easy to see that any phylogenetic network with $r$ reticulate nodes can be extended into a tree-child network by attaching at most $r-1$ leaves labeled with new taxa. Taken together,  the two facts imply that for any set  of trees on $\calX$, there is a tree-child network on a taxa set $\calX' \supseteq \calX$ that displays all the given trees as spanning trees over $\calX$.

2.) The tree-child network in Figure~\ref{fig:ExNet1} has exactly 
$2n-3$ cherries. Let $N$ be a tree-child network on $\calX$ that displays exactly $2n-3$ cherries. Since each cherry has two taxa, each taxa appears on average in $4-6/n$ cherries. This implies that there exists a taxa 
$x$ that appears in  $4$ or more cherries and thus the network displays at most $2n-3-4$ cherries in $\calX\setminus \{x\}$. In summary, 
there is no tree-child network on a $n$-taxa set that displays $2n-3$  cherries in total and also displays  $2n-5$ cherries when restricted to each $(n-1)$-taxa $\mathcal{X}'$ of $\mathcal{X}$.

\subsubsection{Larger tree structures}
\label{sectLargerTrees}

So far we have focused on cherries in the input trees. We now consider larger topological structures in $\mathcal{T}$. We define $k$-subtree to be the subtree of some $T \in \mathcal{T}$ where the number of taxa in the subtree is $k$, for $k \ge 2$. That is, $k$-subtree is a rooted binary tree with $k$ leaves. Note that a $k$-subtree must match topologically exactly a subtree rooted at a node of some $T \in \mathcal{T}$. We do not allow removal of taxa when considering $k$-subtrees. A cherry is a $2$-subtree. 

We first consider the case of $3$-subtree. There can be $3{n \choose 3} = O(n^3) $ \emph{distinct} $3$-subtrees. Proposition \ref{prop3Tree} shows that, to allow a tree-child network, the maximum number of distinct $3$-subtrees is quadratic to $n$.

\begin{proposition}
\label{prop3Tree}
If $\mathcal{T}$ allows a tree-child network, then the maximum number of distinct $3$-subtrees in $\mathcal{T}$ is $O(n^2)$, which is asymptotically tight.
\end{proposition}

\begin{proof}
By Theorem \ref{theoremCherry}, there are at most $2n-3$ distinct cherries in $\mathcal{T}$. Note that each $3$-subtree contains exactly one cherry. For each distinct cherry, we can form a distinct $3$-subtree by adding a third taxon. There are at most $n-2$ choices for the third taxon and there are at most three distinct $3$-subtrees can be formed by each third taxon. Therefore, the maximum number of distinct $3$-subtrees is:

\[  2 (2n-3)(n-2)  = O(n^2)  \]
 
We now show that the $O(n^2)$ bound is asymptomatically tight. We use the same network in Fig. \ref{fig:ExNet1}. 
In this network, there is a $3$-subtree $((1,a),b)$ and also a $3$-subtree $((a,n),b)$ for each $2 \le a < b \le n-1$. Thus there are at least $(n-2)(n-3) = O(n^2)$ distinct $3$-subtrees. 
\end{proof}

\noindent {\bf Remarks}
One can consider larger $k$. By following the same logic, for any $k \ge 3$,  we can build a distinct $k$-subtree from a distinct $k-1$-subtree and then adding one more taxon. There are $2k-3$ ways of inserting the last taxon. 
We can obtain similar bound for any $k$-subtree.
If $\mathcal{T}$ allows a tree-child network, for any $k \ge 3$, it can be shown that there are at most $(2k-3)!! (2n-3) (n-2)(n-3) \ldots (n-k+1)$ distinct $k$-subtrees.
When $k$ increases, this bound in increases quickly and becomes less useful unless the number of trees is large.

\subsection{Four-taxa condition}

As shown in Section \ref{sectLargerTrees}, the number of distinct of $k$-subtrees (say $4$-subtrees) doesn't provide a very strong necessary condition for the existence of tree-child network. We now show that a simple pattern of $3$-subtrees can be used to determine whether a tree-child network exists for $\mathcal{T}$.

\begin{theorem}
\label{theoremFourTaxa}
[Four-taxa condition] 
Let $\calT$ be a set of trees. Suppose there are the following $3$-subtrees in $\mathcal{T}$: $((a,b),d)$, $((a,c),d)$ and $((b,c),d)$ for some $\{a,b,c,d\}$. Then there is no tree-child network that can co-display $\mathcal{T}$. 
\end{theorem}

\begin{proof}
Suppose that $\calT$ is a set of trees satisfying the four-taxa condition and co-displayed in a tree-child network $N$.  
By Lemma \ref{lemmaSingleton}, we can assume  $N$ to have  $n$ non-complex components, which are each a path ending with a leaf. 
For convenience, the component containing a leaf $\ell$ is denoted by $P_\ell$. 
Recall that the $n$ components of $N$ can be ordered topologically such that each reticulate edge goes from left to right. 
By symmetry, we assume that the components $P_a, P_b, P_c$ are listed from left to right. Since there are cherries $(a,c)$ and $(b,c)$, the two reticulate edges going into $P_c$ must be from $P_a$ and $P_b$. Also there is a reticulate edge from  $P_a$ to the component $P_b$. We prove by cases the impossibility of placing $P_d$ in the topological order.

\begin{enumerate}
    
    \item[Case 1.] The  $P_d$ appears after $P_b$. Then there are reticulate edges from  $P_a$ and  $P_b$ into  $P_d$ to form the specified three $3$-subtrees over $a,b,c$ and $d$. Also the source node (within the component $P_a$) of reticulate  edge from  $P_a$ to  $P_d$ must be above the source node of the reticulate edge from $P_a$ to $P_c$. However, such ordering cannot realize the 3-subtree $((a,b),d))$ because $c$ can only be reached by the reticulate edge from $P_a$ to $P_c$ which implies $c$ should be within (not outside) the $3$-subtree $((a,b),d)$.
    
    
    \item[Case 2.] The $P_d$ appears between $P_a$ and  $P_b$. In this case, to form the $3$-subtree $((b,c),d)$, there is a reticulate edge from $P_d$ to $P_b$. That is, there is no more free reticulate edge going into  $P_b$ (and also $P_c$). We consider the relative positions of $\mbox{\rm lca}_N(a, b)$ and $\mbox{\rm lca}_N(a, c)$ (both within $P_a$). First suppose $\mbox{\rm lca}_N(a, b)$ is above $\mbox{\rm lca}_N(a, c)$. Then the $3$-subtree $((a,b),d)$ cannot be obtained because there is no free reticulate edge into $P_c$ (and so $c$ cannot be placed outside this $3$-subtree). Similarly, if $\mbox{\rm lca}_N(a, c)$ is above $\mbox{\rm lca}_N(a, b)$, the the $3$-subtree $((a,c),d)$ cannot be obtained because there is no free reticulate edge into $P_b$.
    
    
    \item[Case 3.] The $P_d$ appear before $P_a$. Since $P_d$ is the leftmost component among the four components under consideration, the root of the $3$-subtree $((a,b),d)$ is located within $P_d$. To display this $3$-subtree, we must follow a reticulate edge from $P_d$ to $P_a$ and a reticulate edge from $P_a$ to $P_b$ (and this excludes the other reticulate edges into $P_a$ and $P_b$). Moreover, reticulate edges into $P_c$ cannot be followed since $c$ is outside this $3$-subtree. However, note that $c$ is only accessible from $P_a$ and $P_b$. Then $c$ cannot be included in the whole tree to display, which leads to a contradiction.
.
\end{enumerate}
\end{proof}

It is easy to construct a set of three trees on $\{1, 2, \cdots,  n\}$ that contains $3$-subtrees $((1,2),4)$, $((1,3),4)$ and $((2,3),4)$ for any $n \ge 4$.  Thus, we have the following fact.

\begin{corollary}
\label{corollary3Trees2}
For any $n\geq 4$, there exist three trees on $n$ taxa that cannot be co-displayed in any tree-child network.
\end{corollary}

\section{Two trees can always be  co-displayed}
\label{sec4}

We now switch to which trees can be co-displayed in a tree-child network. It appears that sufficient conditions are difficult to obtain.
Here,  we show that every pair of trees can always be co-displayed in a tree-child network. 



\begin{theorem}
\label{theorem2Tree}
There exists a tree-child network on $\calX$ to display $T_1$ and $T_2$ for any pair of trees $T_1$ and $T_2$ on $\calX$. 
\end{theorem}


\begin{figure}[t!]
\centering
\includegraphics[scale=1.0]{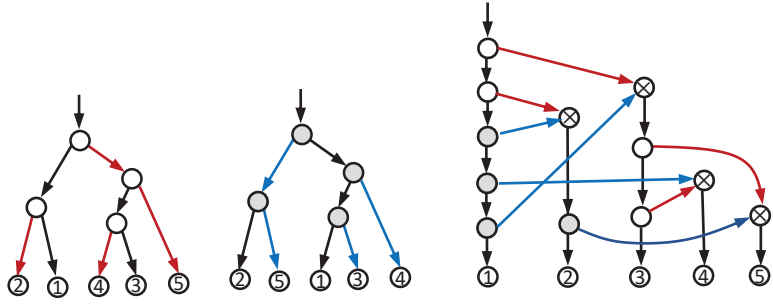}
\caption{Two trees $T_1$ (left), $T_2$ (middle) on four taxa (left) and the tree-child network (right) constructed in Theorem~\ref{theorem2Tree}, which displays the $T_1$ and $T_2$ simultaneously.
In $T_1$ (resp. $T_2$), the red (resp. blue) edges connect different paths of its decomposition. 
In each of the five tree components (vertical paths) of the network, the first node is its reticulate node; the unshaded vertices form the non-trivial paths appearing in the decomposition of $T_1$, while the shaded vertices form the non-trivial paths appearing in the decomposition of $T_2$. The red and blue reticulate edges correspond with the edges connecting different paths in the decomposition of $T_1$ and $T_2$, respectively.
\label{fig7_alg}
}
\end{figure}
\begin{proof}
Let $T_1$ and $T_2$ be two trees on $n$ taxa from $1$ to $n$.
We first decompose  $T_i$ into disjoints paths $P_{ik}$ ($1\leq k\leq n$) for $i=1, 2$ as follows. 
\begin{quote}
    1. $P_{i1}$ is the path consisting of the ancestors of leaf 1, together with the edges between them.
    
    2. For $k=2, \cdots, n$, $P_{ik}$ is the 
       the direct path consisting of the ancestors of leaf $k$ that do not belong to $\cup^{k-1}_{j=1}P_{ij}$ together with the edges between them.
\end{quote}
Let $p_i(k)$ be the parent of leaf $k$ in $T_i$. Note that $P_{i1}$ starts from the root of $T_i$ to $p_i(1)$. For $k\geq 2$, $P_{ik}$ is empty if $p_i(k)$ is in $\cup^{k-1}_{j=1}P_{ij}$ and non-empty otherwise. For example, for $T_1$ in Figure~\ref{fig7_alg}, $P_{11}$ is a 2-node path; $P_{12}$ is empty; $P_{13}$ is a 2-node path; and $P_{14}$ and $P_{15}$ are both empty.
We construct a tree-child network $N$ on $1$-$n$ with $n-1$ reticulate nodes (i.e. $n$ non-complex tree components) as follows. 

The first  component $Q_1$ of $N$ is obtained by connecting
$P_{11}$, $P_{21}$ and leaf 1 by edges (Figure~\ref{fig7_alg}). For $k>1$, 
the $k$-th component $Q_k$  is the concatenation of
a reticulate node $r_k$, $P_{1k}$, $P_{2k}$ and leaf $k$. Moreover, we connect the node that corresponds with the parent of 
the first node of $P_{ik}$ or $p_i(k)$ (if $P_{ik}$ is empty) to $r_k$ using (red or blue) edges for $i=1$  and $2$. In Figure~\ref{fig7_alg}, the red and blue reticulate edges are  added according to the path decomposition of $T_1$ and $T_2$, respectively.  

Since the edges not within a tree component are oriented from a node of  a tree component containing a leaf $i$ to the reticulate node of another tree component containing a leaf $j$ such that $i<j$, the resulting network is acyclic.  
It is easy to  see that the network is also tree-child. Moreover, 
 $T_1$ is obtained from $N$ if blue edges are removed and $T_2$ if red edges are removed.
\end{proof} 






\section{The tree-child hybridization number of multiple trees} 
\label{sec5}

The proof of Theorem~\ref{theorem2Tree} can be generalized to prove the existence of a tree-child network that co-display multiple trees in which each reticulate node is of indegree 2 or more. The minimum reticulate number of the tree-child networks that co-display a set of multiple trees is called their {\it tree-child hybridization number} \cite{van2022practical}. 

\begin{theorem}
\label{tcnumber}
The tree-child hybridization number always exists for any set of multiple trees.
\end{theorem}
\begin{proof}
Let $T_1, T_2, \cdots, T_m$ be a set of $m$ trees on $n$ taxa from $1$ to $n$.
We first decompose  $T_i$ into disjoints paths $P_{ik}$ ($1\leq k\leq n$) for each $i\leq m$ as follows. 
\begin{quote}
    1.) $P_{i1}$ is the path consisting of the ancestors of leaf 1 (including the root of $T_i$), together with the edges between them.
    
    2.)  For $k=2, \cdots, n$, $P_{ik}$ is the 
       the direct path consisting of the ancestors of leaf $k$ that do not belong to $\cup^{k-1}_{j=1}P_{ij}$ together with the edges between them.
\end{quote}
Let $p_i(k)$ be the parent of the leaf $k$ in $T_i$. Note that $P_{i1}$ starts from the root of $T_i$ to $p_i(1)$. For $k\geq 2$, $P_{ik}$ is empty if $p_i(k)$ is in $\cup^{k-1}_{j=1}P_{ij}$ and non-empty otherwise, as shown in Figure~\ref{fig7_alg}.
We construct a tree-child network $N$ on the taxa from $1$ to $n$ with at most $n-1$ reticulate nodes (i.e. $n$ non-complex tree components) as follows. 

The first  component $Q_1$ of $N$ is obtained by connecting
$P_{11}, P_{21}, \cdots, P_{m1}$ and the leaf 1 by edges (Figure~\ref{fig7_alg}). For $k>1$, 
the $k$-th component $Q_k$  is the concatenation of
a reticulate node $r_k$, $P_{1k}, P_{2k}, \cdots, P_{mk}$ and the leaf $k$. Moreover, we connect the node that corresponds with the parent of 
the first node of $P_{ik}$ or $p_i(k)$ (if $P_{ik}$ is empty) to $r_k$ using horizontal edges for each $i\leq m$. (In Figure~\ref{fig7_alg}, the red and blue reticulate edges are  added according to the path decomposition of $T_1$ and $T_2$, respectively).  

Since the edges not within a tree component are oriented from a node of  a tree component containing a leaf $i$ to the reticulate node of another tree component containing a leaf $j$ such that $i<j$, the resulting network is acyclic.  
It is easy to  see that the network is also tree-child. 
\end{proof} 

\noindent {\bf Remark} 
The indegree of each reticulate node in the train-child network constructed in the proof of Theorem~\ref{tcnumber} is $m$.
Thus, the number of the reticulate nodes of the network is unlikely the tree-child hybridization number of the trees.

\section{Discussions and conclusion}
We have contributed a few results on the co-display of multiple trees in tree-child networks. First, we have proved that there are at most $2n-3$ cherries in trees that can be co-displayed in a tree-child network on $n$ taxa. Additionally, we have proved that there is a taxon that is seen only in one or two cherries in the co-displayed trees. We also present several generalizations of these necessary conditions for co-displaying trees in a tree-child network. Second, we have proved that any two trees can always be co-displayed in a binary tree-child network.
and any set of multiple trees can always be co-displayed in a non-binary tree-child network.



\subsection{Algorithmic questions on tree-child networks}
Our study raises the problem of finding a tree-child network with the fewest reticulate nodes among those that can display two given trees. It is known that finding the most parsimonious \emph{unconstrained} phylogenetic networks for displaying two trees is NP complete \cite{SEMPLENPC07}. The computational complexity of finding the most parsimonious tree-child networks for two trees is unknown.
Another interesting algorithmic question is developing fast approaches for constructing parsimonious non-binary tree-child networks for co-displaying a given set of trees. For this problem, we are only aware of one existing approach \cite{van2022practical}, which appears to be slow for large data. 

The tree component decomposition is developed to study the tree containment problem for phylogenetic networks with the reticulation-visibility property 
\cite{Gunawan_16_IC}. It has be used to count galled networks and tree-child networks recently \cite{Cardona_20_JCSS,gunawan_galled}.
This study presents another application of the technique to tree-child networks. Can the technique be used to develop a polynomial time algorithm for determining the co-display of multiple trees in a tree-child network?  

\subsection{Mathematical properties of tree-child networks}
We have showed that for any integer $n$, there are a number of sets of three trees on $n$ taxa that cannot be co-displayed in a tree-child network. This leads to the question of estimating the percentage of 3-tree sets that cannot be co-displayed in any tree-child network. Another question is finding non-trivial sufficient conditions for the existence of tree-child networks with three or more trees.

\section*{Acknowledgments}
This work is partly supported by U.S. National Science Foundation grants CCF-1718093 and IIS-1909425 (to YW) and Singapore MOE Tier 1 grant R-146-000-318-114 (to LZ). The work was started while YW was visiting the Institute for Mathematical Sciences of National University of Singapore in April 2022, which was supported by grant R-146-000-318-114.

%
\bibliographystyle{elsarticle-harv}
\bibliography{Counting_My_references}

\end{document}